\documentclass{llncs}

\usepackage{calc}
\usepackage[nomessages]{fp}
\usepackage{enumitem}
\usepackage{cite}
\usepackage{color}
\usepackage[pdftex]{graphicx}
\usepackage{amssymb}
\usepackage{amsmath}

\newcommand{\simplealgorithm}{Simple Algorithm}

\newcommand{\Topt}[2]{T_{{\scriptsize \textup{OPT}}}(#1,#2)}
\newcommand{\Topts}{T_{\scriptsize \textup{OPT}}}   
\newcommand{\dist}[4][]{\if!#1!T_{#2}(#3,#4)\else T_{#2}(#3,#4,#1)\fi}
\newcommand{\TRV}[3]{T_{{\scriptsize \textup{RV}}}(#1,#2,#3)}

\newcommand{\startpos}[1]{s_{#1}}
\newcommand{\starta}{\startpos{A}}
\newcommand{\startb}{\startpos{B}}
\newcommand{\nat}{\mathbb{N}_{+}}
\newcommand{\cG}{\mathcal{G}}
\newcommand{\cA}{\mathcal{A}}
\newcommand{\agentVariable}{K}

\newcommand{\st}{\hspace{0.1cm}\bigl|\bigr.\hspace{0.1cm}}

\begin{document}
\pagestyle{headings}  
\title{Rendezvous of Heterogeneous Mobile Agents in Edge-weighted Networks
\thanks{Research partially supported by the Polish National Science Center grant DEC-2011/02/A/ST6/00201
and by the ANR project DISPLEXITY (ANR-11-BS02-014). 
This study has been carried out in the frame of the ``Investments for 
the future'' Programme IdEx Bordeaux – CPU (ANR-10-IDEX-03-02).
Dariusz Dereniowski was partially supported by a scholarship for outstanding young researchers funded by the Polish Ministry of Science and Higher Education.}}

\institute{Department of Algorithms and System Modeling, Gda\'nsk University of Technology, Poland.
\and
LaBRI, CNRS and University of Bordeaux, France.
\and
LIAFA, University Paris Diderot, France}

\author{Dariusz Dereniowski\inst{1}%
\and 
Ralf Klasing\inst{2}%
\and 
Adrian Kosowski\inst{3} 
\and
{\L}ukasz Kuszner\inst{1}}

\authorrunning{D.Dereniowski \and  R.Klasing \and A.Kosowski \and {\L}.Kuszner} 
\maketitle

\begin{abstract}
We introduce a variant of the deterministic rendezvous problem for a 
pair of heterogeneous agents operating in an undirected graph, which 
differ in the time they require to traverse particular edges of the 
graph. Each agent knows the complete topology of the graph and 
the initial positions of both agents. The agent also knows its 
own traversal times for all of the edges of the graph, but is unaware of 
the corresponding traversal times for the other agent. The goal of the 
agents is to meet on an edge or a node of the graph. In this scenario, we 
study the time required by the agents to meet, compared to the meeting 
time $\Topts$ in the offline scenario in which the agents have complete 
knowledge about each others speed characteristics. When no additional 
assumptions are made, we show that rendezvous in our model can be 
achieved after time $O(n \Topts)$ in a $n$-node graph, and that such 
time is essentially in some cases the best possible. However, we prove 
that the rendezvous time can be reduced to $\Theta (\Topts)$ when the 
agents are allowed to exchange $\Theta(n)$ bits of information at the 
start of the rendezvous process. We then show that under some natural 
assumption about the traversal times of edges, the hardness of the 
heterogeneous rendezvous problem can be substantially decreased, both in 
terms of time required for rendezvous without communication, and the 
communication complexity of achieving rendezvous in time $\Theta (\Topts)$.
\end{abstract}

\section{Introduction}

Solving computational tasks using teams of agents deployed in a network gives rise to many problems of coordinating actions of multiple agents. Frequently, the communication capabilities of agents are extremely limited, and the exchange of large amounts of information between agents is only possible while they are located at the same network node.
In the rendezvous problem, two identical mobile agents, initially located in two nodes of a network, move along links from node to node, with the goal of occupying the same node at the same time. Such a question has been studied in various models, contexts and applications~\cite{opac-b1117985}.

In this paper we focus our attention on heterogeneous agents in networks, where the time required by an agent to traverse an edge of the network  depends on the properties of the traversing agent. In the most general case we consider, the traversal time associated with every edge and every agent operating in the graph may be different. Scenarios in which traversal times depend on the agent are easy to imagine in different contexts. In a geometric setting, one can consider a road connection network, with agents corresponding to different types of vehicles moving in an environment. One agent may represent a typical road vehicle which performs very well on paved roads, but is unable to traverse other types of terrain. By contrast, the other agent may be a specialized mobile unit, such as a vehicle on caterpillars or an amphibious vehicle, which is able to traverse different types of terrain with equal ease, but without being capable of developing a high speed. In a computer network setting, agents may correspond to software agents with different structure, and the transmission times of agents along links may depend on several parameters of the link being traversed (transmission speed, transmission latency, ability to handle data compression, etc.). 

In general, it may be the case that one agent traverses some links faster than the other agent, but that it traverses other links more slowly. We will also analyze more restricted cases, where we are given some a priori knowledge about the structure of the problem. Specially, we will be interested in the case of \emph{ordered agents}, i.e., where we assume that one agent is always faster than the other one, and the case of \emph{ordered edges}, where we assume that if in a fixed pair of links, one agent takes more time to traverse the first link, the same will also be true for the other agent.

We study the rendezvous problem under the assumption that each agent knows the complete topology of the graph and its traversal times for all edges, but knows nothing about the traversal times or the initial location of the other agent. In all of the considered cases, 
we will ask about the best possible time required to reach rendezvous, 
compared to that in the ``offline scenario'', in which each of the agents also has complete knowledge of the parameters of the other agent. 
We will also study how this time can be reduced by allowing the agents to communicate (exchange a certain number of bits at a distance) 
at the start of the rendezvous process.
 
\subsection{The model and the problem}

Let us consider a simple graph $G = (V,E)$ and its weight functions $w_A : E \mapsto \nat$ and $w_B : E \mapsto \nat$, where $\nat$ is the set of positive integers.
Let $\starta, \startb \in V$, $\starta \neq \startb$, be two distinguished nodes of $G$ -- the agents' $A$ and $B$ starting nodes.
We assume that initially an agent $\agentVariable\in\{A,B\}$ knows the graph $G$, $\starta$, $\startb$ and $w_{\agentVariable}$.
Thus, $A$ knows $w_A$ but it does not know $w_B$, and $B$ knows $w_B$ but it does not now $w_A$.
We assume that the nodes of $G$ have unique identifiers and that $G$ is given to each agent together with the identifiers.
The latter in particular implies that the agents have unique identifiers -- they can `inherit' the identifiers of the nodes $\starta$ and $\startb$.
Also, the agents do not see each other unless they meet.

The weight functions indicate the time required for $A$ and $B$ to move along edges.
That is, given an edge $ e = \{u, v\}$, an agent ${\agentVariable}\in\{A,B\}$ needs $w_{\agentVariable}(e)$ units of time to move along $e$ (in any direction).
We assume that both agents start their computation at time $0$ by exchanging messages. The time required to send and to receive a message is negligible.

Once an agent ${\agentVariable}\in\{A,B\}$ is located at a node $v$, it can do one of the following \emph{actions}:
\begin{itemize}
 \item the agent can wait $t\in \nat$ units of time at $v$; after time $t$ the agent will decide on performing another action,
 \item the agent can start a movement from $v$ to one of its neighbors $u$; in such case the agent moves with the uniform speed from $v$ to $u$ along the edge $\{u,v\}$ and after $w_{\agentVariable}(\{v,u\})$ units of time ${\agentVariable}$ arrives at $u$ and then performs its next action.
\end{itemize}
While an agent is performing its local computations preceding an action, it has access to all messages sent by the other agent at time $0$.
We assume that the time of agent's computations preceding an action is negligible.

We say that $A$ and $B$ \emph{rendezvous at time $t$} 
(or simply \emph{meet}) if they share the same location at time $t$,
\begin{itemize}
 \item they both are located at the same node at time $t$, or
 \item ${\agentVariable}\in\{A,B\}$ started a movement from $u_{\agentVariable}$ to $v_{\agentVariable}$ at time $t_{\agentVariable}<t$, $u_A=v_B$, $v_A=u_B$, $e=\{u_A,v_A\}$, $t_{\agentVariable}+w_{\agentVariable}(e)<t$ and $\frac{t-t_A}{w_A(e)}=1-\frac{t-t_B}{w_B(e)}$ (informally speaking, the agents `pass' each other on $e$ as they start from opposite endpoints of $e$), or
 \item ${\agentVariable}\in\{A,B\}$ started a movement from $u$ to $v$ at time $t_{\agentVariable}<t$, $e=\{u,v\}$, $t_{\agentVariable}+w_{\agentVariable}(e)<t$ and $\frac{t-t_A}{w_A(e)}=\frac{t-t_B}{w_B(e)}$ (informally speaking, both agents start at the same endpoint but the one of them `catches up' the other: $t_A<t_B$ and $w_A(e)>w_B(e)$, or $t_A>t_B$ and $w_A(e)<w_B(e)$).
\end{itemize}

Observe that the last case is not possible in an optimum offline solution, as the agents could rendezvous earlier in the vertex $u$.

\noindent
We are interested in the following problem:
\begin{itemize}
 \item[] Given two integers $b$ and $t$, does there exist an algorithm whose execution by $A$ and $B$ guarantees that the agents send to each other at time $0$ messages consisting of at most $b$ bits in total, and $A$ and $B$ meet after time at most $t$?
\end{itemize}

Given an algorithm for the agents, we refer to the total number of bits sent between the agents as the \emph{communication complexity} of the algorithm.
The \emph{rendezvous time} of an algorithm is the minimum time length $t$ such that the agents meet at time $t$ as a result of the execution of the algorithm.

\subsection{Related work}
The rendezvous problem has been thoroughly studied in the literature in different contexts.
In a general setting, the rendezvous problem was first mentioned in \cite{schelling60}. 
Authors investigating  rendezvous (cf.\cite{opac-b1117985} for an extensive survey) considered either the geometric scenario (rendezvous in an interval of the real line, see, e.g.,  \cite{baston98,BasG01,gal99}, 
or in the plane, see, e.g., \cite{anderson98a,anderson98b}) or the graph scenario (see, e.g.,  \cite{DessmarkFKP06,FP08,KM}). 
A natural extension of the rendezvous problem is that of gathering \cite{FPSW05,IJ90,lim96,thomas92}, 
when more than two agents have to meet in one location.

\paragraph{Rendezvous in anonymous graphs.}

In the anonymous graph model, the agents rely on local knowledge of the graph topology, only. Nodes have no unique identifiers, and maintain only a local labeling of outgoing edges (ports) leading to their neighbors. When studying the feasibility and efficiency of deterministic rendezvous in anonymous graphs, a key problem which needs to be resolved is that of breaking symmetry. Without resorting to marking nodes, this can be achieved by taking advantage of the different labels of agents  \cite{DGKKP,DessmarkFKP06,KM}.  Labeled agents allowed to mark nodes using whiteboards were considered in \cite{YuY96}. Rendezvous of labeled agents using variants of Universal Exploration Sequences was also investigated in~\cite{TSZ07,KM} in the synchronous model, who showed that such meeting can be achieved in time polynomial in the number of nodes of the graph and in the length of the smaller of the labels of the agents. For the case of unlabeled agents, rendezvous is not always feasible when the agents move in synchronous rounds and are allowed only to meet on nodes. However, for any feasible starting configuration, rendezvous of anonymous agents can be achieved in polynomial time, and even more strongly, using only logarithmic memory space of the agent~\cite{CKP12}. In the asynchronous scenario, it has recently been shown that agents can always meet within a polynomial number of moves if they have unique labels~\cite{DPV}. For the case of anonymous agents, the class of instances for which asynchronous rendezvous is feasible is quite similar to that in the synchronous case, though under the assumption that agents are also allowed to meet on edges (which appears to be indispensable in the asynchronous scenario), certain configurations with a mirror-type symmetry also turn out to be gatherable~\cite{GP11}.

\paragraph{Location-aware rendezvous.}
The anonymous scenario may be sharply contrasted with the case in which the agent has full knowledge of the map of the environment, and knows its position within it. Such assumption, partly fueled by the availability and the expansion of the Global Positioning System (GPS), is sometimes called the {\em location awareness} of agents or nodes of the network. Thus, the only unknown variable is the initial location of the other agent. In \cite{leszek-async,DISC10} the authors study the rendezvous problem of location-aware agents in the asynchronous case. The authors of~\cite{leszek-async} introduced the concept of covering sequences that permitted location aware agents to meet along the route of polynomial length in the initial distance $d$ between the agents for the case of multi-dimensional grids. Their result was further advanced in~\cite{DISC10}, where the proposed algorithm provides a route, leading to rendezvous, of length being only a polylogarithmic factor away from the optimal rendezvous trajectory. The synchronous case of location-aware rendezvous was studied in~\cite{CCGKM11}, who provided algorithms working in linear time with respect to the initial distance $d$ for trees and grids, also showing that for general networks location-aware rendezvous carried a polylogarithmic time overhead with respect to $n$, regardless of the initial distance $d$.

\paragraph{Problems for heterogenous agents.}
Scenarios with agents having different capabilities have been also studied.
In \cite{CzyzowiczKP13} the authors considered multiple colliding robots with different velocities traveling along a ring with a goal to determine their initial positions and velocities.
Mobile agents with different speeds were also studied in the context of patrolling a boundary, see e.g. \cite{CzyzowiczGKK11,KawamuraK12}.
In \cite{abs-1304-7693} agents capable of traveling in two different modes that differ with maximal speeds were considered in the context of searching a line segment.
We also mention that speed, although very natural, is not the only attribute that can be used to differentiate the agents.
For example, authors in \cite{Chalopin0MPW13} studied robots with different ranges or, in other words, with different battery sizes limiting the distance that a robot can travel.

\subsection{Additional notation}
Let $\dist[w]{\agentVariable}{u}{v}$, ${\agentVariable}\in\{A,B\}$, 
denote the minimum time required by agent ${\agentVariable}$ to move from $u$ to $v$ in $G$ with a weight function $w$.
If $w=w_{\agentVariable}$, then we write $\dist{{\agentVariable}}{u}{v}$ in place of $\dist[w_{\agentVariable}]{\agentVariable}{u}{v}$, 
${\agentVariable}\in\{A,B\}$.
In other words $\dist{\agentVariable}{u}{v}$ is the length of the shortest path from $u$ to $v$ in $G$ with weight function $w_{\agentVariable}$, where the length of a path composed of edges $e_1,\ldots,e_l$ is $\sum_{j=1}^lw_{\agentVariable}(e_j)$.
We use the symbol $\Topt{\starta}{\startb}$ to denote the minimum time for rendezvous in the off-line setting where agents that are initially placed on $\starta$ and $\startb$ know all parameters. We will skip starting positions if it will not lead to confusion writing simply $\Topts$.
Denote also $M_{\agentVariable}:=\max\{w_{\agentVariable}(e)\st e\in E\}$, ${\agentVariable}\in\{A,B\}$, and let $M:=\max\{M_A,M_B\}$.
All logarithms have base $2$, i.e., we write for brevity $\log$ in place of $\log_2$.

The following lemma, informally speaking, implies that we do not have to consider scenarios in which rendezvous occurs on edges, and by doing so we restrict ourselves to solutions among which there exists one that is within a constant factor from an optimal one.
Let $\TRV{\starta}{\startb}{v}$ denote the minimum time for rendezvous at $v$, that is, $\TRV{\starta}{\startb}{v}=\max \{ \dist{A}{\startpos{A}}{v}, \dist{B}{\startpos{B}}{v} \}$.
Let any node $u$ that minimizes the $\TRV{\starta}{\startb}{u}$ be called a \emph{rendezvous node}.

\begin{lemma} \label{lem:meet_on_nodes}
For each graph $G=(V,E)$ and for each $\starta,\startb\in V$, if $u\in V$ is the rendezvous node, then $\TRV{\starta}{\startb}{u} \leq 2\Topt{\starta}{\startb}$.
\end{lemma}
\begin{proof}
If the two agents can achieve rendezvous on a node in time $\Topt{\starta}{\startb}$, then the lemma follows and hence we assume in the following that rendezvous occurs on an edge.
For ${\agentVariable}\in\{A,B\}$, let $v_{\agentVariable}$ be the last node visited by ${\agentVariable}$ prior to rendezvous that the two agents achieve in time $\Topt{\starta}{\startb}$.
Observe that $v_A\neq v_B$ and $e=\{v_A,v_B\}\in E$.

In an optimum solution at least one of the agents traversed at least half of $e$, so
\begin{equation} \label{eq:rv_node:1}
	2\Topt{\starta}{\startb} \geq
	\min \{ \dist{A}{\starta}{v_B}, \dist{B}{\startb}{v_A} \}.
\end{equation}
Moreover,  $\dist{A}{\starta}{v_B} > \dist{B}{\startb}{v_B}$ and $\dist{B}{\startb}{v_A} > \dist{A}{\starta}{v_A}$ , so
\begin{eqnarray}
  \min \{ \dist{A}{\starta}{v_B}, \dist{B}{\startb}{v_A} \}
	& =&
	\min \{
		\max \{ \dist{A}{\starta}{v_B}, \dist{B}{\startb}{v_B}\}, \nonumber\\
	& &	\hphantom{\min \{} \max \{ \dist{B}{\startb}{v_A}, \dist{A}{\starta}{v_A}\}
	\} \label{eq:rv_node:2}\\
	&=&
	\min \{ \TRV{\starta}{\startb}{v_B}, \TRV{\starta}{\startb}{v_A} \}. \nonumber
\end{eqnarray}
If $u$ is a rendezvous node, then
\[
  	\min \{ \TRV{\starta}{\startb}{v_B}, \TRV{\starta}{\startb}{v_A} \} \geq \TRV{\starta}{\startb}{u}.
\]
This, \eqref{eq:rv_node:1} and \eqref{eq:rv_node:2} prove the lemma.
\qed\end{proof}

\subsection{Possible restrictions on weight functions}

Arbitrary weight functions might cause very bad performance of rendezvous (see Theorems~\ref{thm:arbitrary+comm+lower} and~\ref{thm:case1+no-comm}).
Thus, beside the arbitrary case, we will be interested in restricted cases, namely:
\begin{enumerate}
\item \label{arbitraryW} $w_A$ and $w_B$ are \emph{arbitrary functions},
\item \label{monotoneEdges}  $\forall_{e_1, e_2 \in E} \,\, w_A(e_1) < w_A(e_2) \iff w_B(e_1) < w_B(e_2)$,
\item \label{monotoneAgents} $\forall_{e \in E} \,\, w_A(e) \leq w_B(e)$ or $\forall_{e \in E} \,\, w_B(e) \leq w_A(e)$.
\end{enumerate}

Case~\ref{arbitraryW} reflects the situation where both agents and edges are not related in terms of time needed to move along them. 
Whenever two functions have the property case~\ref{monotoneEdges}, we will refer to the problem instance as the case of \emph{ordered edges}.
Informally, in such scenario both agents obtain the same ordering of edges (up to resolving ties) with respect to their weights.
The last case reflects the situation where one of the agents is always at least as fast as the other one.
Instances with this property are referred to as the cases of \emph{ordered agents}.

\subsection{Our results}

In this work we analyze the following two extreme scenarios.
In the first scenario (the middle column in Table~\ref{tab:summary}) we consider the communication complexity of algorithms that guarantee that rendezvous occurs in time $\Theta(\Topts)$ regardless of the starting positions.
In the second scenario (the third column) we provide bounds on the rendezvous time in case when the agents send no messages to each other.

\begin{table}[htb]
\begin{center}
\caption{Summary of results ($n$ is the number of nodes of the input graph)}\label{tab:summary}
\begin{tabular}{|p{3.25cm}||p{4.75cm}|p{3.5cm}|} \hline
& communication complexity for \newline
	rendezvous in time $\Theta(\Topts)$
& rendezvous time in case  \newline
 of no communication \\ \hline\hline

Case~\ref{arbitraryW}: arbitrary 
&
$O(n \cdot (\log \log (M \cdot n)))$ (Thm.~\ref{thm:arbitrary+comm_compl}) \newline
$\Omega(n)$ (Thm.~\ref{thm:arbitrary+comm+lower})
 &
$\Theta(n\cdot\Topts)$ (Thms~\ref{thm:arbitrary+no-comm},~\ref{thm:case1+no-comm})

\\ \hline

Case~\ref{monotoneEdges}: ordered edges &
$O(\log\log M+\log^2n)$ (Thm.~\ref{thm:monotone_edges+no-visibility}) \newline
$\Omega(\log n)$ (Thm.~\ref{thm:case2+comm+lower})
&
$O(n\cdot\Topts)$ (Thm.~\ref{thm:arbitrary+no-comm}) \newline
$\Omega(\sqrt{n}\cdot\Topts)$ (Thm.~\ref{thm:case2+no-comm})
\\  \hline

Case~\ref{monotoneAgents}: ordered agents
&  none (Thm.~\ref{thm:orderedagents})
& $\Theta(\Topts)$ (Thm.~\ref{thm:orderedagents}) \\ \hline
\end{tabular}
\end{center}
\end{table}

\section{Communication complexity for $\Theta(\Topts)$ time} \label{sec:comm-compl}
In this section we determine upper and lower bounds for communication complexity of algorithms that achieve rendezvous in asymptotically optimal time.
Section~\ref{sec:comm-compl} is subdivided into three parts reflecting the three cases of weight functions we consider.

\subsection{The case of arbitrary functions}
We start by giving an upper bound on communication complexity of asymptotically optimal rendezvous.
Our method is constructive, i.e., we provide an algorithm for the agents (see proof of Theorem~\ref{thm:arbitrary+comm_compl}).
Then, (cf. Theorem~\ref{thm:arbitrary+comm+lower}) we give the corresponding lower bound.

\begin{theorem} \label{thm:arbitrary+comm_compl}
There exists an algorithm 
that guarantees rendezvous in  $\Theta(\Topts)$ time 
and has communication complexity $O(n(\log\log (M \cdot n)))$ for arbitrary functions. 
\end{theorem}
\begin{proof}
Let $I_0=[0,1]$, and for $j>0$ let $I_j=(2^{j-1},2^j]$.
Denote $V=\{v_1,\ldots,v_n\}$, where the vertices are ordered according to their identifiers.
We first formulate an algorithm and then we prove that it has the required properties.
We assume that $A$ is the executing agent and $B$ is the other agent (the algorithm for $B$ is analogous).
\begin{enumerate}
 \item For each $j=1,\ldots,n$ (in this order) send to $B$ the integer $r(A,j)$ such that $\dist{A}{\startpos{A}}{v_j}\in I_{r(A,j)}$.
 \item After receiving the corresponding messages from $B$, construct $T'\colon V\mapsto \nat$ such that
   \[T'(v_j):=\max \{2^{r(A,j)}, 2^{r(B,j)}\}, \quad j\in\{1,\ldots,n\}.\]
 \item Find a node $v_{\rho}$ with minimum value of $T'(v_{\rho})$.
  If more than one such node $v_{\rho}$ exists, then take $v_{\rho}$ to be the one with minimum identifier.
 \item Go to $v_{\rho}$ along a shortest path and stop.
\end{enumerate}

Note that both agents compute the same function $T'$.
This in particular implies that the same vertex $v_{\rho}$, to which each agent goes, is selected by both agents.
Hence, the agents rendezvous at $v_{\rho}$.
The transmission of $r({\agentVariable},j)$ requires $O(\log\log (M \cdot n))$ bits because $r({\agentVariable},j)=O(\log (M \cdot n))$ 
for each ${\agentVariable}\in\{A,B\}$ and $j\in\{1,\ldots,n\}$.
Thus, the communication complexity of the algorithm is $O(n\log\log (M \cdot n))$.

\smallskip
We now give an upper bound on the rendezvous time at $v_{\rho}$. By definition,  for each $j \in\{1,\ldots,n\}$ 
and for each ${\agentVariable}\in\{A,B\}$ we have
\begin{equation*}
2^{r({\agentVariable},j)-1} \leq \dist{\agentVariable}{s_{\agentVariable}}{v_j} \leq 2^{r({\agentVariable},j)}.
\end{equation*}
Thus, having in mind that $\TRV{\starta}{\startb}{v}=\max\{\dist{A}{\starta}{v},\dist{B}{\startb}{v}\}$, we obtain:
\begin{equation} \label{eq:Tprim}
 \frac{1}{2} T'(v_j)\leq \TRV{\starta}{\startb}{v_j} \leq T'(v_j), \quad j\in\{1,\ldots,n\}.
\end{equation}
Now, let $u$ be a rendezvous node. By \eqref{eq:Tprim},  the choice of index $\rho$, again by \eqref{eq:Tprim} and by
Lemma~\ref{lem:meet_on_nodes} we obtain
\begin{equation*}
\TRV{\starta}{\startb}{v_{\rho}} \leq T'(v_{\rho}) \leq T'(u) \leq 2\TRV{\starta}{\startb}{u} \leq
4\Topt{\starta}{\startb},
\end{equation*}
which completes the proof.
\qed\end{proof}

\begin{theorem} \label{thm:arbitrary+comm+lower}
Each algorithm that guarantees rendezvous in time $\Theta(\Topts)$ has communication complexity $\Omega(n)$ for some $n$-node graphs.
\end{theorem}

\begin{proof}
Let $\cG$ be a class of graph such that each $G\in\cG$ is a complete bipartite graph  $K_{2, n}$
with $V = \{\starta, \startb, v_1, v_2, \ldots v_n\}$
and $E=E_A\cup E_B$, where
$E_{\agentVariable} = \left\{ \{\startpos{{\agentVariable}}, v_j\} \st j \in \{1, 2, \ldots n \} \right\},\quad {\agentVariable}\in\{A,B\}$,
and, for each ${\agentVariable}\in\{A,B\}$, $w_{\agentVariable}(e)=X$ for each $e\in E\setminus E_{\agentVariable}$ and $w_{\agentVariable}(e) \in \{1, X\}$ for each $e\in E_{\agentVariable}$,
where $X$ is a sufficiently big integer, e.g., $X=n$.

Note that for each $G\in\cG$, $\Topts\in\{1,X\}$.
Moreover, $\Topts=1$ if and only if there exists an index $j\in\{1,\ldots,n\}$ such that $w_A(\{\starta, v_j\})$ = $w_B(\{\startb, v_j\}) = 1$.
A problem to find such an index $j$ is equivalent to a known problem of set intersection~\cite{Kalyanasundaram:1992:PCC:140820.140844}
and requires $\Omega(n)$ bits to be transmitted between $A$ and $B$.
\qed\end{proof}

\subsection{The case of ordered edges}

\begin{theorem} \label{thm:monotone_edges+no-visibility}
There exists an algorithm that guarantees rendezvous in  $\Theta(\Topts)$ time and 
has communication complexity $O(\log\log M+\log^2n)$ in case of monotone edges.
\end{theorem}
\begin{proof}
Let $I_0=[0,1]$, and for $j>0$ let $I_j=(2^{j-1},2^j]$.
For ${\agentVariable}\in\{A,B\}$ and a function $w_{\agentVariable}\colon E \mapsto \nat$ let $m(w_{\agentVariable})$ be the maximum integer such that the removal of all edges from $G$ with weights greater than $m(w_{\agentVariable})$ disconnects $G$ in such a way that $\starta$ and $\startb$ belong to different connected components.
For ${\agentVariable}\in\{A,B\}$ and $j\geq 0$, define $r^{\agentVariable}_j=|\{e\in E\st w_{\agentVariable}(e)\in I_j\}|$.

We now give a statement of an algorithm with communication complexity $O(\log\log M+\log^2n)$.
Then, we prove that its execution by each agent guarantees rendezvous in time $\Theta(\Topts)$.
\begin{enumerate}
 \item Let $A$ be the executing agent and let $B$ be the other agent (the statement for $B$ is analogous).
  Send to $B$ the index $c_A$ such that $m(w_A)\in I_{c_A}$ (this requires sending $\log c_A\leq\log\log m(w_A)=O(\log\log M)$ bits).
  Set $c:=\min\{c_A,c_{B}\}$ ($c_{B}$ is in the corresponding message received from $B$).
 \item Send to $B$ the value of $r^A_{c}$ and, for each $j\in\{1,\ldots,\lceil\log n\rceil\}$ send to $B$ the values of $r^A_{c+j}$ and $r^A_{c-j}$ (this requires sending $O(\log^2n)$ bits in total).
 \item Send to $B$ the value of $r^A:=r^A_0+r^A_1+\cdots+r^A_{c-\lceil\log n\rceil-1}$ (this requires sending $O(\log n)$ bits).
 \item After receiving the corresponding messages from $B$
  construct a weight function $\widetilde{w}_{B}\colon E \mapsto \nat$ as follows.
  First, sort the edges so that $w_A(e_j)\leq w_A(e_{j+1})$ for each $j\in\{1,\ldots,|E|-1\}$.
  Denote
   \[E_0^{B}=\{e_1,\ldots,e_{r^{B}}\} \quad\textup{and}\quad E_{\infty}^{B}=E\setminus\{e\st \widetilde{w}_{B}(e)\in I_{0}\cup\cdots\cup I_{c+\lceil\log n\rceil}\}.\]
  Then, $\widetilde{w}_{B}(e):=0$ for each $e\in E_0^{B}$; $\widetilde{w}_{B}(e)=+\infty$ for each $e\in E_{\infty}^{B}$; and for each edge $e\in E\setminus(E_0^B\cup E_{\infty}^B)$ set
  $\widetilde{w}_{B}(e):=2^{j'-1}$ if $e\in I_{j'}$ (this can be deduced from messages received from $B$).
 \item Calculate the function $\widetilde{w}_{A}$ (i.e., the function that $B$ constructs based on the information sent to $B$).
 \item Find a node $v_{\rho}\in V$ such that $\max\{\dist[\widetilde{w}_{A}]{A}{\startpos{A}}{v_{\rho}},\dist[\widetilde{w}_{B}]{B}{\startpos{B}}{v_{\rho}}\}$ is minimum.
  If more than one such node exists, then take $v_{\rho}$ to be the one with minimum identifier.
 \item Go to $v_{\rho}$ along a shortest path and stop.
\end{enumerate}
Note that the communication complexity of the above algorithm is $O(\log\log M+\log^2n)$.
Also, both agents calculate $\widetilde{w}_{A}$ and $\widetilde{w}_{B}$ and hence the node $v_{\rho}$ is the same for both agents, which implies that the algorithm guarantees rendezvous.

Therefore, it remains to prove that 
\[
 \max\{\dist{A}{\startpos{A}}{v_{\rho}},\dist{B}{\startpos{B}}{v_{\rho}}\}=O(\Topt{\startpos{A}}{\startpos{B}}).
\]
Due to Lemma~\ref{lem:meet_on_nodes}, it is enough to show that
\begin{equation} \label{eq:meet_on_nodes}
\max\{\dist{A}{\startpos{A}}{v_{\rho}},\dist{B}{\startpos{B}}{v_{\rho}}\}=O(\TRV{\starta}{\startb}{u}),
\end{equation}
where $u$ is a rendezvous node.
For ${\agentVariable}\in\{A,B\}$ and $x\in\{u,v_{\rho}\}$, let $P_{\agentVariable}^x$ be the set of edges of a shortest path from $\startpos{{\agentVariable}}$ to $x$ in $G$ with weight function $w_{\agentVariable}$ and let $\widetilde{P}_{\agentVariable}^x$ be the set of edges of a shortest path from $\startpos{{\agentVariable}}$ to $x$ in $G$ 
with weight function $\widetilde{w}_{\agentVariable}$.

Note that \eqref{eq:meet_on_nodes} follows from
\begin{equation} \label{eq:meet_on_nodes2}
\max\left\{w_A(P_A^{v_{\rho}}),w_B(P_B^{v_{\rho}})\right\}=O\left(\max\{w_A(P_A^u),w_B(P_B^u)\}\right).
\end{equation}
Hence, we focus on proving the latter equation.

First note that for each ${\agentVariable}\in\{A,B\}$ and $X\subseteq E$ it holds: 
\begin{equation} \label{eq:zero-edges}
w_{\agentVariable}(X\cap E_0^{{\agentVariable}})\leq |X|\cdot 2^{c-\lceil\log n\rceil-1}\leq \frac{|X|}{n}\cdot 2^{c-1} \leq \frac{|X|}{n}\cdot\min\{m(w_A),m(w_B)\}.
\end{equation}
Thus, for each ${\agentVariable}\in\{A,B\}$,
\begin{align} \label{eq:zero-edges:2}
\begin{split}
w_{\agentVariable}(\widetilde{P}_{\agentVariable}^{v_{\rho}}\cap E_0^{\agentVariable}) & \leq \frac{|\widetilde{P}_{\agentVariable}^{v_{\rho}}\cap E_0^{\agentVariable}|}{n}\cdot \min\{m(w_A),m(w_B)\} \\
& < \min\{m(w_A),m(w_B)\}
\end{split}
\end{align}
because the edges in $\widetilde{P}_{\agentVariable}^{v_{\rho}}$ 
form a path which in particular gives $|\widetilde{P}_{\agentVariable}^{v_{\rho}}\cap E_0^{\agentVariable}|<n$.
Inequality \eqref{eq:zero-edges:2} implies that, for each ${\agentVariable}\in\{A,B\}$,
\begin{equation} \label{eq:bound:v:1}
w_{\agentVariable}(\widetilde{P}_{\agentVariable}^{v_{\rho}}) \leq \min\{m(w_A),m(w_B)\} + w_{\agentVariable}(\widetilde{P}_{\agentVariable}^{v_{\rho}}\setminus E_0^{\agentVariable}).
\end{equation}

Note that
\begin{equation} \label{eq:no_big_edges}
P_{\agentVariable}^u\cap E_{\infty}^{\agentVariable} = \emptyset, \quad {\agentVariable}\in\{A,B\}.
\end{equation}
Indeed, otherwise $w_{\agentVariable}(P_{\agentVariable}^u)>2^{c+\lceil\log n\rceil}\geq n2^{c}$ contradicting that $u$ is a rendezvous node (the agent ${\agentVariable}$ with $c_{\agentVariable}=c$ can reach the starting position of the other agent in time not greater than 
$(n-1)2^{c_{\agentVariable}}=(n-1)2^{c}$).

Denote $Y_{\agentVariable}=\{e\in E \st w_{\agentVariable}(e)>2^{c-1}\}$ for each ${\agentVariable}\in\{A,B\}$.
We prove that
\begin{equation} \label{eq:some_bid_edge_is_there}
 Y_A\cap P_A^u\neq\emptyset \quad\textup{or}\quad Y_B\cap P_B^u\neq\emptyset.
\end{equation}
Suppose for a contradiction that $Y_A\cap P_A^u=\emptyset$ and $Y_B\cap P_B^u=\emptyset$.
Assume without loss of generality that $|Y_A|\leq|Y_B|$ (the analysis in the opposite case is analogous).
By definition of $Y_A$, $w_A(e)\leq 2^{c-1}$ for each $e\in P_A^u$.
Since we consider the case of monotone edges, $|Y_A|\leq|Y_B|$ implies that $w_A(e)\leq 2^{c-1}$ for each $e\in P_B^u$.
This however means that a subset of edges in $P_A^u\cup P_B^u$ gives a path whose each edge $e$ satisfies $w_A(e)\leq 2^{c-1}\leq 2^{c_A-1}$, contradicting the choice of $c_A$.
This completes the proof of \eqref{eq:some_bid_edge_is_there}.
Note that \eqref{eq:some_bid_edge_is_there} implies
\begin{equation} \label{eq:min:bound}
\min\{m(w_A),m(w_B)\} \leq 2^c \leq 2 \max\left\{ w_A(P_A^u),w_B(P_B^u) \right\}.
\end{equation}

By the definition of $P_{\agentVariable}^{v_{\rho}}$ and by \eqref{eq:bound:v:1} we have 
\begin{equation} \label{eq:newbound:1}
w_{\agentVariable}(P_{\agentVariable}^{v_{\rho}}) \leq w_{\agentVariable}(\widetilde{P}_{\agentVariable}^{v_{\rho}}) 
\leq w_{\agentVariable}(\widetilde{P}_{\agentVariable}^{v_{\rho}}\setminus E_0) + \min\{m(w_A),m(w_B)\}.
\end{equation}
Recall that  ${w_{\agentVariable}}(e) \leq 2\widetilde{w}_{\agentVariable}(e)$ 
for each $e\in E \setminus E_0^{\agentVariable}$ and ${\agentVariable}\in\{A,B\}$.
Thus,
\begin{equation} \label{eq:newbound:2}
w_{\agentVariable}(\widetilde{P}_{\agentVariable}^{v_{\rho}}\setminus E_0) \leq  
2\widetilde{w}_{\agentVariable}(\widetilde{P}_{\agentVariable}^{v_{\rho}}\setminus E_0).
\end{equation}
Since $\widetilde{w}_{\agentVariable}(e)=0$ for each $e\in E_0^{\agentVariable}$ and ${\agentVariable}\in\{A,B\}$,
\begin{equation} \label{eq:newbound:3}   
\widetilde{w}_{\agentVariable}(\widetilde{P}_{\agentVariable}^{v_{\rho}}\setminus E_0) 
=
\widetilde{w}_{\agentVariable}(\widetilde{P}_{\agentVariable}^{v_{\rho}}).
\end{equation}
Combining \eqref{eq:newbound:1}, \eqref{eq:newbound:2} and \eqref{eq:newbound:3} we obtain 
\begin{equation} \label{eq:importantbound:1}
w_{\agentVariable}(P_{\agentVariable}^{v_{\rho}}) \leq 2 \widetilde{w}_{\agentVariable}(\widetilde{P}_{\agentVariable}^{v_{\rho}}) + \min\{m(w_A),m(w_B)\}.
\end{equation}
By definition of $\widetilde{P}_{\agentVariable}^{u}$ and \eqref{eq:no_big_edges},
\begin{equation} \label{eq:importantbound:2}
\widetilde{w}_{\agentVariable}(\widetilde{P}_{\agentVariable}^{u}) \leq
\widetilde{w}_{\agentVariable}(P_{\agentVariable}^{u}) \leq
w_{\agentVariable}(P_{\agentVariable}^{u}), \quad \agentVariable\in\{A,B\}.
\end{equation}
By \eqref{eq:importantbound:1}, the choice of $v_{\rho}$, \eqref{eq:importantbound:2} and \eqref{eq:min:bound},
\begin{gather*}
\max\left\{w_A(P_A^{v_{\rho}}),w_B(P_B^{v_{\rho}})\right\} \leq 
\hspace*{6cm}\phantom{123} \\
\begin{align}
& \leq \min\{m(w_A),m(w_B)\} + 2\max\left\{\widetilde{w}_{A}(\widetilde{P}_A^{v_{\rho}}),\widetilde{w}_{B}(\widetilde{P}_B^{v_{\rho}})\right\} \nonumber \\
& \leq \min\{m(w_A),m(w_B)\} + 2\max\left\{\widetilde{w}_{A}(\widetilde{P}_A^{u}),\widetilde{w}_{B}(\widetilde{P}_B^{u})\right\} \nonumber \\
& \leq \min\{m(w_A),m(w_B)\} + 2\max\left\{w_A(P_A^{u}),w_B(P_B^{u})\right\} \nonumber \\
& \leq 4\max\left\{ w_A(P_A^u),w_B(P_B^u) \right\}. \nonumber
\end{align}
\end{gather*}
This proves \eqref{eq:meet_on_nodes2} and completes the proof of the theorem.
\qed\end{proof}

\begin{theorem} \label{thm:case2+comm+lower}
Each algorithm that guarantees rendezvous in time $\Theta(\Topts)$ has communication complexity $\Omega(\log n)$ for some $n$-node graphs in case of ordered edges.
\end{theorem}

\begin{proof}
Let $k$ be a positive integer.
We first define a family of graphs $\cG=\{G_1,\ldots,G_k\}$.
Each graph in $\cG$ has the same structure, namely, the vertices $\starta$ and $\startb$ are connected with $k$ node-disjoint paths but the graphs in $\cG$ have different weight functions associated with them.
Each of those paths consists of exactly $k+2$ edges (see Figure~\ref{fig:fig_case2}).
\begin{figure}[htb]
\begin{center}
\includegraphics[scale=0.95]{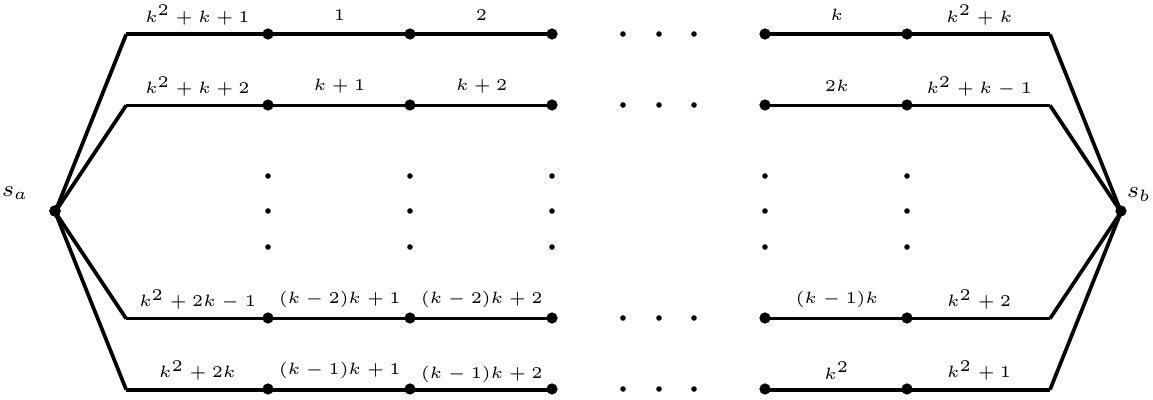}
\caption{The structure of the graphs in the proof of Theorems~\ref{thm:case2+comm+lower} and~\ref{thm:case2+no-comm}; the numbers give ordering of edges with respect to agents' weight functions}
\label{fig:fig_case2}
\end{center}
\end{figure}

Thus, each graph in $\cG$ has $k^2+k+2$ nodes and $m=k^2+2k$ edges.
The edges are denoted by $e_1,\ldots,e_{m}$.
The location of each edge in $\cG$ is shown in Figure~\ref{fig:fig_case2}, where for improving the presentation we write $i$ in place of $e_i$ for each $i\in\{1,\ldots,m\}$.
We will set the labels of the edges so that $w_{\agentVariable}(e_{1}) < \cdots < w_{\agentVariable}(e_{m})$ for each ${\agentVariable} \in \{A, B\}$.
Now, for each graph in $\cG$, we put 
\begin{equation*}
w_A(e_i) := X + i \quad \textup{for each }i\in\{1,\ldots,m\},
\end{equation*} 
where $X=k^4$.
For $i\in\{1,\ldots,m\}$ and $j\in\{1,\ldots,k\}$, we set the weight function $w_B$ for $G_j$ as follows:
\begin{equation*}
  w_B(e_i) := 
	\begin{cases}
  i,  & \text{for } i \leq jk, \\
	X + i, & \text{for } jk <i \leq k^2 + k - j + 1, \\
  kX + i, & \text{for } k^2 + k - j + 1 < i. \\  
	\end{cases}
\end{equation*}
Note that $w_A(e_1)<\cdots<w_B(e_m)$ and $w_B(e_1)<\cdots<w_B(e_m)$ in each graph $G_j\in\cG$ which ensures that all problem instances are cases of ordered edges.

Let $G_j\in\cG$.
Denote by $H_1,\ldots,H_k$ the $k$ edge-disjoint paths connecting $\starta$ and $\startb$, where $H_{j'}$ is the path containing the edge $e_{k^2+k+j'}$ incident to $\starta$, $j'\in\{1,\ldots,k\}$.
We argue that if $A$ and $B$ rendezvous on a path $H_{j'}$ in time at most $kX/2$, then $j'=j$.
First note that $A$ is not able to reach any vertex adjacent to $\startb$ in time $kX/2$.
Also, $j'<j$ is not possible for otherwise $B$ would traverse one of the edges $e_{k^2+k-j+2},\ldots,e_{k^2+k}$, each of weight at least $kX$ --- a contradiction.
Now, suppose for a contradiction that $j'>j$.
Then, one of the agents traverses at least half of the path $H_{j'}$, i.e., it traverses at least $k/2+1$ of its edges.
If this agent is $A$, then clearly rendezvous occurs not earlier than $(k/2+1)X$ --- a contradiction.
If this agent is $B$, then it does not traverse any of the edges $e_{1},\ldots,e_{jk}$, since those belong to paths $H_1,\ldots,H_j$ and we have $j'>j$.
Hence, by the definition of $w_B$, we also have rendezvous after more than $kX/2$ time units, which gives the required contradiction.
We have proved that, in $H_j$, rendezvous is obtained before time $kX/2$ only if it occurs on the path $H_j$.

Observe that for $n$ large enough it holds $\Topt{\starta}{\startb} < 2X$ for each $G_j\in\cG$.
To see that, let $A$ traverse the edge $e_{k^2+k+j}$, and let $B$ traverse the remaining edges of $H_j$, i.e., $e_{(j-1)k+1},e_{(j-1)k+2},\ldots,e_{(j-1)k+k}$ and $e_{k^2+k-j+1}$.
We have $w_A(e_{k^2+k+j})=X+\Theta(k^2)$ and
\begin{align}
 w_B(e_{k^2+k-j+1})+\sum_{l=1}^{k} w_B(e_{(j-1)k+l}) & = X+k^2+k-j+1 + \sum_{l=1}^k((j-1)k+l) \nonumber \\
   & = X+ O(k^3). \nonumber
\end{align}
Since $X=k^4$, we obtain that $\Topt{\starta}{\startb}<2X$ for $n$ large enough.

Suppose for a contradiction that there exists an algorithm $\cA$ that guarantees rendezvous in time $\Theta(\Topts)$ and has communication complexity $o(\log n)$.
Let $C$ be such a constant that the rendezvous time guaranteed by ${\cA}$ is bounded by $C \Topts$.
We will show that for $n>C^4$, the algorithm $\cA$ that sends at most $C_2\log n$ bits, where $C_2 = \frac{3}{8}$, cannot guarantee rendezvous in time $C\Topts$, which will give the desired contradiction.

Note that in each algorithm, and thus in particular in $\cA$, if the agent $A$ receives the same message from $B$ in two different graphs in $\cG$, then $A$ must traverse the same sequence of edges for both graphs.

The number of all possible messages that $A$ might receive using $C_2 \log n$ bits is at most $2^{C_2 \log n}$.
Observe that $2^{C_2 \log n} < \frac{k}{C}$. Indeed, 
\[
  C_2 = \frac{3}{8} = \frac{\log n^{3/4}}{2 \log n} < \frac{\log (n/C)}{2 \log n} < \frac{\log (k/C)}{\log n}  
\] 
implies the required inequality.
Thus, there must exist $\cG'$, a subset of $\cG$, with at least $C+1$ elements such that $A$ traverses the same sequence of edges for all graphs in $\cG'$. 

To traverse an edge adjacent to $\starta$ from $H_i$, the agent $A$ requires $X+i > X$ units of time.
In order to traverse any of the remaining edges, $A$ requires to traverse back the mentioned edge first using $X+i > X$ units of time again.
Hence, during $C\Topts < 2CX$ time, $A$ is able to traverse at most $C$ edges adjacent to $\starta$. 
It implies that there must exist an index $j$ such that $G_j \in \cG'$ and agent $A$ does not traverse an edge from $H_j$ adjacent to $\starta$. 
It means that $A$ has not met $B$ in time $2CX  > C\Topts$, a contradiction.   
\qed\end{proof}

\subsection{The case of ordered agents}
Now we will present a general solution which achieves $\Theta(\Topt{\starta}{\startb})$ time without communication for the case of ordered agents.
This property allows us to obtain our asymptotically tight bounds both for communication complexity of optimal-time rendezvous and for optimal rendezvous time with no communication.
We point out that, unlike in previous cases, the algorithms for the agents are different, i.e., $A$ and $B$ perform different (asymmetric) actions.
We assume that the algorithm for the agent $A$ (respectively, $B$) is executed by the agent whose starting node has smaller (bigger, respectively) identifier.
Note that, since both agents know the graph and both starting nodes, they can correctly decide on executing an algorithm.

\subsubsection*{\simplealgorithm}

Tasks of agent $A$:
\begin{enumerate}
  \item wait $\dist{A}{\starta}{\startb}$ units of time,
  \item  go to $\startb$ along an arbitrarily chosen shortest path (according to the weight function $w_A$) from $\starta$ to $\startb$, and return to $\starta$ along the same path and stop.
\end{enumerate}
\noindent
Tasks of agent $B$:
\begin{enumerate}
 \item  wait $\dist{B}{\starta}{\startb}$ units of time,
 \item go to $\starta$ along an arbitrarily chosen shortest path (according to the weight function $w_B$) from $\startb$ to $\starta$ and stop.
\end{enumerate}

\begin{lemma} \label{lem:verysimpleprotocol}
For the case of ordered agents, {\simplealgorithm} guarantees rendezvous in time $6 \min \{ \dist{A}{\starta}{\startb}, \dist{B}{\starta}{\startb}\}$.
\end{lemma}
\begin{proof}	
For sure, $A$ and $B$ will eventually rendezvous, as both of them reach $\starta$ and stay there.
Let $y$ be the time point at which the agents rendezvous.
Let us consider agent $A$. It might meet $B$ while:
\begin{enumerate}
	\item waiting $\dist{A}{\starta}{\startb}$ units of time at $\starta$. In this case $y = 2\dist{B}{\starta}{\startb}  \leq  \dist{A}{\starta}{\startb}$.
	\item moving towards $\startb$ or on the way back to $\starta$.
              Clearly $y \leq  3\dist{A}{\starta}{\startb}$.
              Also, agent $B$ at time point $y$ is either at $\startb$ or is moving from $\startb$ to $\starta$.
              Thus, $y \leq 2\dist{B}{\starta}{\startb}$.

	\item arriving at $\starta$, i.e., rendezvous occurs at $\starta$ at the moment when $A$ returns to $\starta$.
              Clearly, $y\leq 3\dist{A}{\starta}{\startb}$.
              As the agents have not met at $\starta$ before $A$ started moving, we have $\dist{A}{\starta}{\startb}\leq 2\dist{B}{\starta}{\startb}$.
	      So, $y\leq 6T_B(\starta,\startb)$.

	\item waiting for $B$ at $\starta$ after the path traversals.
              Clearly, $y=2\dist{B}{\starta}{\startb}$.
              In this case, as the agents have not met at $\startb$, we have $\dist{B}{\starta}{\startb} \leq 2\dist{A}{\starta}{\startb}$.
	      So, $y\leq 4\dist{A}{\starta}{\startb}$.
	
\end{enumerate}
\qed\end{proof}

We remark that the constant $6$ from Lemma~\ref{lem:verysimpleprotocol} might be reduced to $2\sqrt{2}+3$
if we would allow both agents $A$ and $B$ to wait a little longer in the initial state:
$\sqrt{2} \dist{A}{\starta}{\startb}$ and $\sqrt{2} \dist{B}{\starta}{\startb}$ respectively.

\begin{lemma} \label{lem:ordered-agents}
In the case of ordered agents we have
\[ \min \{ T_A(\starta,\startb), \dist{B}{\starta}{\startb} \} \leq 2 \Topt{\starta}{\startb}
\]
\end{lemma}	
\begin{proof}
Suppose that both agents rendezvous at $x$ after $\Topt{\starta}{\startb}$ units of time.
If rendezvous does not occur at a node, then with a slight abuse of notation we write $\dist{{\agentVariable}}{u}{x}$ to denote the time an agent ${\agentVariable}$ needs to go from a node $u$ to $x$. 
Suppose without loss of generality that $A$ is a `faster' agent, i.e., $w_A(e)\leq w_B(e)$ for each edge $e$.
This in particular implies that $\dist{A}{\starta}{\startb}\geq \dist{B}{\starta}{\startb}$ and hence it remains to provide the upper bound on $\dist{A}{\starta}{\startb}$.
Moreover, by first using the triangle inequality we have
$\dist{A}{\starta}{\startb} \leq T_A(\starta, x) + T_A(x, \startb) \leq T_A(\starta, x) + T_B(x, \startb) \leq 2\Topt{\starta}{\startb}$.
\qed\end{proof}

Now, due to Lemmas~\ref{lem:verysimpleprotocol} and~\ref{lem:ordered-agents}, we are ready to conclude:
\begin {theorem} \label{thm:orderedagents}
In the case of ordered agents (case~\ref{monotoneAgents}) there exists an algorithm that guarantees rendezvous in time $\Theta(\Topt{\starta}{\startb})$
without performing any communication. \qed
\end{theorem}

\section{Rendezvous with no communication}

\subsection{The case of arbitrary functions}
Theorem~\ref{thm:arbitrary+no-comm} below gives the upper bound on rendezvous time without communication.
Then, Theorem~\ref{thm:case1+no-comm} provides our lower bound for this case.

\begin{theorem} \label{thm:arbitrary+no-comm}
There exists an algorithm that without performing any communication guarantees rendezvous in time $O(n\cdot\Topt{\starta}{\startb})$, where $n$ is the number of nodes of the network.
\end{theorem}
\begin{proof}
We start by giving an algorithm.
Its first step in an initialization and the remaining steps form a loop.
Denote $V=\{v_1,\ldots,v_n\}$.
\begin{enumerate}
 \item\label{st:1} Let initially $x:=1$. Let ${\agentVariable}$ be the executing agent.
 \item\label{st:2} For each $j\in\{1,\ldots,n\}$ do:
   \begin{enumerate}[label={\normalfont \ref{st:2}.\arabic*.}]
    \item If $\dist{\agentVariable}{\startpos{{\agentVariable}}}{v_j}\leq x$, then set $x':=\dist{\agentVariable}{\startpos{\agentVariable}}{v_j}$ and go to $v_j$ along a shortest path. Otherwise, set $x':=0$.
    \item Wait $x-x'$ time units at the current node.
    \item Return to $\startpos{{\agentVariable}}$ along a shortest path. (This step is vacuous if $x'=0$.)
    \item Wait $x-x'$ time units.
   \end{enumerate}
 \item\label{st:3} Set $x:= 2x$ and return to Step~\ref{st:2}.
\end{enumerate}

Let us introduce some notation regarding the above algorithm.
We divide the time into \emph{phases}, where the $p$-th phase, $p\geq 0$, consists of all time units in which both agents were performing actions determined in Step~\ref{st:2} for $x=2^{p}$.
Then, each phase is further subdivided into \emph{stages}, where the $s$-th stage, $s\in\{1,\ldots,n\}$, of the $p$-th phase consists of all time units in which both agents were performing actions determined in Step~\ref{st:2} for $x=2^{p}$ and $j=s$.
Note that these definitions are correct since both agents simultaneously start at time $0$.

First observe, by a simple induction on the total number of stages, that at the beginning of each stage each agent ${\agentVariable}\in\{A,B\}$ is present at $\startpos{{\agentVariable}}$.
We now prove that both agents are guaranteed to rendezvous at a rendezvous node $v$ in the $p$-th phase, where $2^p\geq \max\{\dist{A}{\starta}{v},\dist{B}{\startb}{v}\}$.
Consider the $s$-th stage of $p$-th phase such that $v_s=v$.
Since $2^p\geq \max\{\dist{A}{\starta}{v},\dist{B}{\startb}{v}\}$, both agents reach $v$ in at most $2^p$ moves.
Due to the waiting time of $2^p-\dist{{\agentVariable}}{\startpos{{\agentVariable}}}{v}$ of agent ${\agentVariable}\in\{A,B\}$ after reaching $v$, we obtain that both agents are present at $v$ at the end of the $2^p$-th time unit of the $s$-stage in the $p$-th phase.
This completes the proof of the correctness of our algorithm.

It remains to bound the time in which the agents rendezvous.
The duration of the $p$-th phase is $O(n2^p)$.
The total number of phases is at most $P=\lceil \log \max\{\dist{A}{\starta}{v},\dist{B}{\startb}{v}\} \rceil$.
Thus, the agents rendezvous in time
\[
O(n\sum_{p=1}^P2^p)=O(n2^P)=O\left(n\cdot \max\{\dist{A}{\starta}{v},\dist{B}{\startb}{v}\} \right).
\]
Lemma~\ref{lem:meet_on_nodes} implies that the agents rendezvous in time $O(n\cdot\Topt{\starta}{\startb})$ as required.
\qed
\end{proof}

\begin{theorem} \label{thm:case1+no-comm}
Any  algorithm that without performing any communication guarantees rendezvous uses time $\Omega(n\cdot\Topt{\starta}{\startb})$, where $n$ is the number of nodes of the network.
\end{theorem}

\begin{proof}
Let us consider the complete bipartite graph $G$ given in the proof of Theorem~\ref{thm:arbitrary+comm+lower}
with $V(G) = \{\starta, \startb, v_1, v_2, \ldots v_n\}$
and $E=E_A\cup E_B$, where
\[E_{\agentVariable} = \left\{ \{\startpos{{\agentVariable}}, v_j\} \st j \in \{1, 2, \ldots n \} \right\},\quad {\agentVariable}\in\{A,B\}.\]

Let $w_A(e)=X$ for each $e\in E_B$ and $w_A(\{\starta, v_i\}) = 1$ for each $e\in E_A$, where $X$ is some sufficiently big integer, say $X=n$.
We will now give a partial definition of $w_B$, starting with $w_B(e)=X$ for each $e\in E_A$.
This weight functions will be constructed in such a way that rendezvous at time $1$ is possible.
Informally, we will set only one edge in $E_B$ to have weight $1$ for the agent $B$ while the remaining edges will have weight $X$.
This is done by analyzing possible moves of the agent $A$.

Now, we consider an arbitrary sequence of moves of agent $A$ during the first $n$ time units.
Clearly, after this time, agent $A$ is not able to reach $s_B$.
There also exists an edge $\{\starta,v_j\}\in E_A$ that agent $A$ performed no move along it, i.e., $A$ did not visit $v_j$.
We set $w_B(\{\starta, v_j\}) := 1$ and $w_B(\{\starta, v_i\}) := X$ for all $i \neq j$.

It is easy to observe that $\Topt{\starta}{\startb}$ is equal to $1$ and this time can be achieved only by a meeting at $v_j$.
However, $A$ and $B$ did not rendezvous during the first $n$ time units.
Thus, there exists no algorithm that guarantees rendezvous in time $o(n \cdot \Topt{\starta}{\startb})$.
\qed\end{proof}

\subsection{Lower bound for the case of ordered edges without communication}

\begin{theorem} \label{thm:case2+no-comm}
In the case of ordered edges, any algorithm that guarantees rendezvous without performing any communication uses time $\Omega(\sqrt{n}\cdot\Topt{\starta}{\startb})$, where $n$ is the number of nodes of the network.
\end{theorem}
\begin{proof}
We will use the same family of graphs $\cG$ as constructed in the proof of Theorem~\ref{thm:case2+comm+lower}; see also Figure~\ref{fig:fig_case2}.
Recall that if $A$ and $B$ rendezvous on a path $H_{j'}$ in time at most $kX/2$, then $j'=j$ and
$\Topts = O(X)$ for each $G_j\in\cG$.

\smallskip
Note that, the agent $A$ has the same input for each graph in $\cG$ since $w_A$ is the same for all graphs in $\cG$.
Thus, for any algorithm $\cA$, the agent $A$ traverses the same sequence of edges for each graph in $\cG$.
Moreover, rendezvous time bounded by $kX/2$ (see the proof of Theorem~\ref{thm:case2+comm+lower}) implies that there exist edges adjacent to $\starta$ that $A$ does not traverse.
In other words, there exists $j\in\{1,\ldots,k\}$ such that $A$ traverses no edge of $H_j$.
Therefore, we obtain that $A$ and $B$ cannot rendezvous in $G_j$ in time less than $kX/2$.
Since $k=\Theta(\sqrt{n})$ and rendezvous can be achieved in time $O(X)$ for each graph in $\cG$ the proof has been completed.
\qed\end{proof}

\section{Final remarks}
It seems that the most interesting and challenging among the analyzed cases is the one of ordered edges without communication.
There is still a substantial gap between the lower and the upper bounds we have provided and we leave it an interesting open question whether there exists an algorithm with a better approximation ratio than that of $O(n\Topts)$.
It is also intersting if the upper bound $M$ on the weights of the edges affects the communication complexity
for arbitrary functions and the cases of ordered edges.

Another interesting research direction is to analyze scenarios in which we allow agents to communicate at any time. 
To point out an advantage that the agents may gain in such case, note that the agents can rendezvous very quickly in graphs that we used for a lower bound in the proof of the Theorem~\ref{thm:case2+comm+lower}.
Indeed, transmitting just one bit in the moment correlated with the index of the preferred (optimal for rendezvous) path would help the agents to learn which path they should follow.

\end{document}